\documentclass[pdflatex,sn-mathphys-num]{sn-jnl}

\usepackage{graphicx}%
\usepackage{multirow}%
\usepackage{amsmath,amssymb,amsfonts}%
\usepackage{amsthm}%
\newtheorem{theorem}{Theorem}[section]
\newtheorem{proposition}{Proposition}[section]
\newtheorem{lemma}{Lemma}[section]
\newtheorem{corollary}{Corollary}[section]

\usepackage{mathrsfs}%
\usepackage[title]{appendix}%
\usepackage{xcolor}%
\usepackage{textcomp}%
\usepackage{manyfoot}%
\usepackage{booktabs}%
\usepackage{algorithm}%
\usepackage{algorithmicx}%
\usepackage{algpseudocode}%
\usepackage{listings}%
\usepackage{natbib}
\usepackage{enumitem}

\theoremstyle{thmstyletwo}%
\newtheorem{remark}{Remark}%

\begin{document}

\title[Success Probability in Shor’s Algorithm]{Success Probability in Shor’s Algorithm}

\author*[1]{\fnm{Ali} \sur{Abbassi}}\email{ali.abbassi@orange.com}
\author*[2]{\fnm{Lionel} \sur{Bayle}}\email{lionel.bayle@univ-angers.fr}

\affil*[1]{\orgdiv{LIST3N}, \orgname{Orange Innovation}, \orgaddress{\country{France}}}
\affil*[2]{\orgdiv{LAREMA, Sfr MATHSTIC}, \orgname{CNRS, University of Angers}, \orgaddress{\city{Angers}, \postcode{F-49000}, \country{France}}}


\abstract{This paper aims to determine the exact success probability at each step of Shor’s algorithm. Although the literature usually provides a lower bound on this probability, we present an improved bound. The derived formulas enable the identification of all failure cases in Shor’s algorithm, which correspond to a success probability of zero. A simulation routine is provided to evaluate the theoretical success probability for a given integer when its prime factorization is known with potential applications in quantum resource estimation and algorithm benchmarking.}


\maketitle

\section{Introduction}\label{sec1}

Shor’s algorithm marked a turning point in quantum computing by offering an exponentially faster method for integer factorization~\cite{Shor1997}. The algorithm relies on the selection of a pseudo-random integer satisfying certain properties and proceeds in three main steps: random selection, quantum order finding, and classical post-processing. A lower bound on the success probability was given in~\cite{Shor1996}, but no closed-form expression or exact value is known.

The goal of this paper is to compute the exact success probability for each of the three steps of Shor’s algorithm, as well as for the full algorithm, across all integers—including those not simply of the form \( N = pq \). We provide these values in closed form for all relevant cases based on the prime factor decomposition of the input. Furthermore, we identify and classify all specific failure cases that may occur in different stages of the algorithm.

We emphasize that our analysis is based on Shor’s original algorithm. However, more efficient classical post-processing methods have since been developed that allow complete factorization from a single order-finding call with high probability~\cite{Ekera2021, Ekera2023}. Implementations of these improved techniques are available in open-source libraries such as \texttt{Quaspy} and \texttt{Factoritall}.

A simulation routine is included for evaluating the theoretical success probability for a given input when its factorization is known, with applications in quantum resource estimation and algorithm benchmarking.

This article is structured as follows. Section~\ref{sec2} introduces the mathematical foundations of Shor’s algorithm. Sections~\ref{sec:step1} to~\ref{sec3} compute the exact success probabilities for each of the three steps. Section~\ref{success-pr} gives the probability of overall success. Section~\ref{sec:pseudocode} describes a simulation procedure based on these results, for use in controlled environments where the factorization of \( N \) is available. The final section~\ref{conclusion} presents the conclusion.

\section{Mathematical Foundations of Shor's Algorithm}\label{sec2}

Shor’s algorithm \cite{Shor1997} is a quantum algorithm that solves the problem of integer factorization exponentially faster than the best-known classical algorithms. Given a composite integer \( N \geq 2 \), the algorithm finds a non-trivial factor of \( N \) by reducing the problem to that of order finding in the multiplicative group \( (\mathbb{Z}/N\mathbb{Z})^* \). More precisely, for a randomly chosen integer \( a \in \{2, \dots, N-1\} \) with \( \gcd(a, N) = 1 \), the algorithm seeks the smallest integer \( r \in \mathbb{N}^* \) such that:
\begin{equation}
    a^r \equiv 1 \pmod{N}. \label{eq:order}
\end{equation}
This integer \( r \) is called the \emph{order} of \( a \) modulo \( N \), and under appropriate conditions, it reveals non-trivial factors of \( N \) with high probability.

The quantum speed-up arises from the efficient computation of \( r \) via the \emph{Quantum Phase Estimation} (QPE) algorithm applied to the unitary operator corresponding to modular exponentiation. This quantum subroutine runs in time \(O((\log N)^3) \) and space \( O(\log N) \), whereas classical algorithms like the Number Field Sieve run in time roughly \(\exp\left(O\left((\log N)^{1/3} (\log \log N)^{2/3}\right)\right)\), which is sub-exponential but super-polynomial in \(\log N\) \cite{Pomerance}
under standard assumptions. The classical parts of the algorithm include random sampling, GCD computations (e.g., via Euclid’s algorithm), and simple divisibility tests.

The full procedure is summarized in Algorithm~\ref{algo-shor}.

\begin{algorithm}[H]
\caption{Shor’s Factoring Algorithm}\label{algo-shor}
\begin{algorithmic}[1]
\Require Composite integer \( N > 1 \)
\Ensure A non-trivial factor of \( N \)
\Repeat
    \State Select random \( a \in \{2, \dots, N-1\} \)
    \State \( d \gets \gcd(a, N) \)
    \If{\( d > 1 \)} \Return \( d \) \EndIf
    \State Use QPE to compute order \( r \) of \( a \mod N \)
    \If{\( r \) odd or \( a^{r/2} \equiv -1 \mod N \)} \textbf{continue} \EndIf
    \State \( x_1 \gets \gcd(a^{r/2} - 1, N) \), \( x_2 \gets \gcd(a^{r/2} + 1, N) \)
    \If{\( 1 < x_1 < N \)} \Return \( x_1 \)
    \ElsIf{\( 1 < x_2 < N \)} \Return \( x_2 \)
    \EndIf
\Until{a factor is found}
\end{algorithmic}
\end{algorithm}

The algorithm succeeds when the order \( r \) is even and satisfies \( a^{r/2} \not\equiv -1 \pmod{N} \). Under these conditions, one has the identity:
\begin{equation}
    (a^{r/2} - 1)(a^{r/2} + 1) \equiv 0 \pmod{N}, \label{eq:factorization}
\end{equation}
which implies that \( N \mid (a^{r/2} - 1)(a^{r/2} + 1) \). As a result, at least one of the values \( \gcd(a^{r/2} \pm 1, N) \) yields a non-trivial factor of \( N \).

The probability that a single iteration of the algorithm yields a factor is closely related to the structure of the multiplicative group \( (\mathbb{Z}/N\mathbb{Z})^* \), particularly the proportion of elements with even order and the behavior of their roots.

The key mathematical insight is that the distribution of orders \( r \) depends on the decomposition of \( (\mathbb{Z}/N\mathbb{Z})^* \) into cyclic subgroups. The following structural results describe this decomposition:

\begin{theorem}[\cite{Rombaldi}]\label{thm:Z2k}
Let \( k \in \mathbb{N}^* \). Then:
\begin{itemize}
    \item \( \left((\mathbb{Z}/2\mathbb{Z})^*, \times\right) \cong (\{0\}, +) \);
    \item \( \left((\mathbb{Z}/4\mathbb{Z})^*, \times\right) \cong (\mathbb{Z}/2\mathbb{Z}, +) \);
    \item For \( k \geq 3 \),
    \[
        \left((\mathbb{Z}/2^k\mathbb{Z})^*, \times\right) \cong \left(\mathbb{Z}/2\mathbb{Z} \times \mathbb{Z}/2^{k-2}\mathbb{Z},+\right).
    \]
\end{itemize}
\end{theorem}

\begin{theorem}[\cite{Rombaldi}]\label{thm:Zpk}
Let \( p \) be an odd prime and \( k \in \mathbb{N}^* \). Then:
\[
    \left((\mathbb{Z}/p^k\mathbb{Z})^*, \times\right) \cong \left(\mathbb{Z}/[(p-1)p^{k-1}]\mathbb{Z},+\right).
\]
\end{theorem}

These isomorphisms allow explicit enumeration of the elements of even order, and thus inform the success probability of each iteration of Shor’s algorithm. In Section~\ref{sec3}, we use these results to quantify the conditions under which the algorithm fails or succeeds.

\section{Step 1}\label{sec:step1}

Let \( N \geq 2 \) be a natural number with prime decomposition given by
\[
N = \prod_{i=0}^{l} p_i^{k_i},
\]
where the \( p_i \) are distinct prime numbers and the \( k_i \) are positive integers.

\begin{proposition}[\cite{Lang}]\mbox{}\\
The number of invertible elements in \( \mathbb{Z}/N\mathbb{Z} \), i.e., the number of integers between \( 1 \) and \( N-1 \) that are coprime to \( N \), is given by Euler’s totient function:
\[
\varphi(N) = \prod_{i=0}^{l} (p_i - 1) p_i^{k_i - 1}.
\]
Consequently, the probability that a uniformly randomly chosen integer in \( \{0, 1, \dots, N-1\} \) is coprime to \( N \) is
\[
\frac{\varphi(N)}{N} = \prod_{i=0}^{l} \left(1 - \frac{1}{p_i}\right) = \prod_{i=0}^{l} \frac{p_i - 1}{p_i}.
\]
\end{proposition}

\section{Step 2}\label{sec:step2}

Step 2 consists in counting the invertible elements of \( \mathbb{Z}/N\mathbb{Z} \) of even order.

\subsection{Case \( N = p_1^{k_1} \cdots p_l^{k_l} \), with \( p_i \) odd primes such that \( p_i - 1 = 2^{s_i} m_i \)}\label{subsec:odd-case}

Throughout this section, we fix \( N = p_1^{k_1} \cdots p_l^{k_l} \), where \( p_i \) are distinct odd primes, \( l \geq 1 \), and \( p_i - 1 = 2^{s_i} m_i \) with \( s_i \in \mathbb{N}^* \), \( m_i \in \mathbb{N}^* \), and \( \gcd(m_i, 2) = 1 \) for each \( i \).

\begin{lemma}[Case \( p \neq 2 \), \( N = p^k \), \( p-1 = 2^s m \), with \( \gcd(m, 2) = 1 \)]\label{lem:even-order-p}
Let \( (p - 1) p^{k - 1} = 2^s m p^{k - 1} = 2^s m' \), so that \( \gcd(m, 2) = \gcd(m', 2) = 1 \). The conditional probability that the class of \( a \) in \( \mathbb{Z}/N\mathbb{Z} \) has even order, given that it is invertible, is
\[
1 - \frac{1}{2^s}.
\]
\end{lemma}

\begin{proof}
The additive group \( \mathbb{Z}/(2^s m')\mathbb{Z} \) is isomorphic, via the Chinese Remainder Theorem, to the product group
\[
\mathbb{Z}/2^s\mathbb{Z} \times \mathbb{Z}/m'\mathbb{Z}.
\]
Its group of units \( (\mathbb{Z}/2^s\mathbb{Z})^* \) is isomorphic to a product of cyclic groups whose orders are powers of 2, hence it contains only one element of odd order: the identity. Since \( m' \) is odd, all elements of \( (\mathbb{Z}/m'\mathbb{Z})^* \) have odd order. 

As the order of an element in a product group is the least common multiple of the orders in the components, the group \( (\mathbb{Z}/(2^s m')\mathbb{Z})^* \) contains exactly \( m' \) elements of odd order. Therefore, the number of invertible elements of even order is \( 2^s m' - m' \), and the corresponding probability is
\[
\frac{2^s m' - m'}{2^s m'} = 1 - \frac{1}{2^s}.
\]
\end{proof}

\begin{proposition}\label{prop:even-order-no2}
The conditional probability that a uniformly sampled integer \( a < N \), coprime to \( N \), has even order in \( \mathbb{Z}/N\mathbb{Z} \) is
\[
1 - \frac{1}{2^{s_1}} \cdots \frac{1}{2^{s_l}}.
\]
\end{proposition}

\begin{proof}
The multiplicative group \( (\mathbb{Z}/N\mathbb{Z})^* \) is isomorphic to the product
\[
(\mathbb{Z}/p_1^{k_1}\mathbb{Z})^* \times \cdots \times (\mathbb{Z}/p_l^{k_l}\mathbb{Z})^*,
\]
which, using previous lemmas, is in turn isomorphic to
\[
\prod_{i=1}^l \mathbb{Z}/2^{s_i}\mathbb{Z} \times \mathbb{Z}/m'_i\mathbb{Z},
\]
where each \( m'_i = m_i p_i^{k_i - 1} \). 

An element in this product group has odd order if and only if each component in the \( \mathbb{Z}/2^{s_i}\mathbb{Z} \) factor is zero. This yields exactly \( \prod_{i=1}^l m'_i \) elements of odd order. Therefore, the probability of drawing an element of even order is
\[
\frac{\prod_{i=1}^l 2^{s_i} m'_i - \prod_{i=1}^l m'_i}{\prod_{i=1}^l 2^{s_i} m'_i} = 1 - \prod_{i=1}^l \frac{1}{2^{s_i}}.
\]
\end{proof}
\subsection{Case \( N = 2^{k_0} p_1^{k_1} \cdots p_l^{k_l} \), with \( p_i \) odd primes such that \( p_i - 1 = 2^{s_i} m_i, k_i\geq1. \)}\label{subsec:even-case}

We now consider the general case where \( N = 2^{k_0} p_1^{k_1} \cdots p_l^{k_l} \), with \( k_0 \geq 1 \), \( l \geq 0 \), and for all \( i \), \( p_i \) are odd primes such that \( p_i - 1 = 2^{s_i} m_i \) with \( \gcd(m_i, 2) = 1 \).

\begin{proposition}\label{prop:even-order-with2}
The probability that an element \( \bar{a} \in (\mathbb{Z}/N\mathbb{Z})^* \) has even order is
\[
\frac{2^{k_0 - 1} \prod_{i=1}^l 2^{s_i} m'_i - \prod_{i=1}^l m'_i}{2^{k_0 - 1} \prod_{i=1}^l 2^{s_i} m'_i} = 1 - \frac{1}{2^{k_0 - 1}} \prod_{i=1}^l \frac{1}{2^{s_i}}.
\]
\end{proposition}

\begin{proof}
We have the isomorphism:
\[
(\mathbb{Z}/N\mathbb{Z})^* \cong (\mathbb{Z}/2^{k_0}\mathbb{Z})^* \times \prod_{i=1}^l (\mathbb{Z}/p_i^{k_i}\mathbb{Z})^*.
\]
An element of this product group has odd order only if all components have odd order. In the factor \( (\mathbb{Z}/2^{k_0}\mathbb{Z})^* \), the order of every non-identity element is divisible by 2, and only the identity has odd order. Thus, this factor contributes only one element of odd order.

For the remaining components, the previous computation applies. Consequently, the total number of invertible elements with odd order is \( \prod_{i=1}^l m'_i \), and the probability that an element of \( (\mathbb{Z}/N\mathbb{Z})^* \) has even order is
\[
\frac{2^{k_0 - 1} \prod_{i=1}^l 2^{s_i} m'_i - \prod_{i=1}^l m'_i}{2^{k_0 - 1} \prod_{i=1}^l 2^{s_i} m'_i} = 1 - \frac{1}{2^{k_0 - 1}} \prod_{i=1}^l \frac{1}{2^{s_i}}.
\]

We verify that this formula remains valid in the case \( l = 0 \): then \( N = 2^{k_0} \), and the group \( (\mathbb{Z}/2^{k_0}\mathbb{Z})^* \) has \( 2^{k_0 - 1} \) elements, among which only the identity has odd order. Hence, the probability that a randomly chosen invertible element has even order is
\[
1 - \frac{1}{2^{k_0 - 1}},
\]
as expected.
\end{proof}

\section{Step 3}\label{sec3}

\begin{proposition}[Case $N = p^k$, $p \neq 2$, $p - 1 = 2^s$, equation $a^2 = 1$]\mbox{}\\
Let $p \neq 2$ be a prime number, $N = p^k$ with $k \geq 1$, and assume $p - 1 = 2^s$, hence $s \geq 1$.\\
Solving the equation $a^2 = 1$ in the group $\left((\mathbb{Z}/p^k\mathbb{Z})^*, \times\right)$ is equivalent to solving $2a = 0$ in the group $\left(\mathbb{Z}/(p-1)p^{k-1}\mathbb{Z}, +\right) \simeq \left(\mathbb{Z}/2^s\mathbb{Z} \times \mathbb{Z}/p^{k-1}\mathbb{Z}, +\right)$, since $2$ and $p$ are coprime.\\
As $p^k$ is odd, the equation $(2x, 2y) = (0,0)$ holds if and only if $x = 2^{s-1}u$ with $u \in \{0,1\}$, since $2$ is invertible in $\mathbb{Z}/p^{k-1}\mathbb{Z}$. Therefore, $x = 0$ or $x = 2^{s-1}$ and $y = 0$. The solution $(0,0)$ corresponds to $1$ in $\left((\mathbb{Z}/p^k\mathbb{Z})^*, \times\right)$, and the solution $(2^{s-1}, 0)$ corresponds to $-1$.
\end{proposition}

We can now identify cases where Shor’s algorithm will never succeed.

\begin{proposition}[Case $N = p^k$, $k \geq 1$, $p \neq 2$, $p - 1 = 2^s$, $s \geq 1$]\mbox{}\\
Let $a$ be an element of even order $r$. As shown above, the equation $x^2 = 1$ has only two solutions: $1$ and $-1$. Since $(a^{r/2})^2 = 1$ and $a^{r/2} \neq 1$ (as $a$ has order $r$), it must be that $a^{r/2} = -1$. Hence, Shor’s algorithm will never succeed in this case.
\end{proposition}

\begin{remark}
The primes described in the previous proposition are precisely the Fermat primes. It is conjectured that only finitely many exist, and currently only five Fermat primes are known.
\end{remark}

\begin{proposition}[Case $N = p^k$, $k \geq 1$, $p \neq 2$, $p - 1 = 2^s m$, $s \geq 1$, $m \geq 3$, equation $a^2 = 1$]\mbox{}\\
Solving the equation $a^2 = 1$ in $\left((\mathbb{Z}/p^k\mathbb{Z})^*, \times\right)$ is equivalent to solving $2a = 0$ in $\left(\mathbb{Z}/(p-1)p^{k-1}\mathbb{Z}, +\right) \simeq \left(\mathbb{Z}/2^s\mathbb{Z} \times \mathbb{Z}/m\mathbb{Z} \times \mathbb{Z}/p^{k-1}\mathbb{Z}, +\right)$, as $2$, $m$, and $p$ are pairwise coprime.\\
Since both $m$ and $p^k$ are odd, the condition $(2x, 2y, 2z) = (0,0,0)$ holds if and only if $x = 2^{s-1}u$ for $u \in \{0,1\}$, $y = 0$, and $z = 0$, since $2$ is invertible in both $\mathbb{Z}/m\mathbb{Z}$ and $\mathbb{Z}/p^{k-1}\mathbb{Z}$. Therefore, $x = 0$ or $x = 2^{s-1}$, $y = 0$, and $z = 0$. The solution $(0,0,0)$ corresponds to $1$ in $\left((\mathbb{Z}/p^k\mathbb{Z})^*, \times\right)$, and the solution $(2^{s-1},0,0)$ corresponds to $-1$.
\end{proposition}

In this case as well, Shor’s algorithm will never succeed.

\begin{proposition}[Case $N = 2^k$, $k \geq 3$, equation $a^2 = 1$]\mbox{}\\
In $\left((\mathbb{Z}/2^k\mathbb{Z})^*, \times\right) \simeq \left(\mathbb{Z}/2\mathbb{Z} \times \mathbb{Z}/2^{k-2}\mathbb{Z}, +\right)$, the equation $a^2 = 1$ has four solutions: $1 \simeq (0,0)$, $5^{2^{k-3}} \simeq (0, 2^{k-3})$, $-1 \simeq (1,0)$, and $-5^{2^{k-3}} \simeq (1, 2^{k-3})$ \cite{Rombaldi}. Shor’s algorithm succeeds if and only if $a^{r/2}$ equals either $5^{2^{k-3}}$ or $-5^{2^{k-3}}$.\\
The probability that a randomly chosen integer $a < N$ coprime to $2$ is $\frac{1}{2}$. Among these, the probability that $a$ has odd order (i.e., corresponds to $(0,0)$) is $\frac{1}{2^{k-1}}$. Hence, the probability that $a$ has even order is $1 - \frac{1}{2^{k-1}}$. Therefore, the probability that a random $a < N$, coprime to $2$, has even order is $\frac{1}{2}\left(1 - \frac{1}{2^{k-1}}\right)$. Finally, the probability that such an $a$ has even order and satisfies $a^{r/2} \neq -1$ is $\frac{1}{2}\left(1 - \frac{1}{2^{k-2}}\right)$.
\end{proposition}

\begin{proof}
Let $a \simeq (x, y)$. Then $a^2 = 1$ if and only if $(2x, 2y) = (0, 0)$, which implies that $x$ is arbitrary and $y = 0$ or $y = 2^{k-3}$. This yields four solutions: $1 \simeq (0,0)$, $5^{2^{k-3}} \simeq (0, 2^{k-3})$, $-1 \simeq (1,0)$, and $-5^{2^{k-3}} \simeq (1, 2^{k-3})$.\\
The case where $r(x, y) = (0,0)$ and $(r/2)(x,y) = (1,0)$ corresponds to $x = 1$, $r$ even, and $r/2$ odd, thus $r = 2$, since $r$ is a power of $2$, and $(r/2)y = 0$ implies $y = 0$. The unique solution is then $(x, y) = (1, 0)$, which corresponds to $-1$ in the group.\\
The group $\left((\mathbb{Z}/2^k\mathbb{Z})^*, \times\right)$ has $2^{k-1}$ elements coprime to $2$, hence invertible. Only one of them has odd order: the identity. Since the group has order a power of $2$, all element orders are powers of $2$. Moreover, only one element $-1$ satisfies $a^{r/2} = -1$, as previously shown. The stated probabilities follow.
\end{proof}

\begin{proposition}[Case $N = 2^k$, $k = 1$ or $2$]
If $N = 2^k$ with $k = 1$, the ring $\mathbb{Z}/2\mathbb{Z}$ contains only one invertible element, namely $1$, which has odd order. Shor’s algorithm halts after the first step.

If $k = 2$, then $\mathbb{Z}/4\mathbb{Z}$ has two invertible elements: $1$ of order $1$ (odd), and $3 = -1$, which has order $2$ and satisfies $3^1 = -1$. Shor’s algorithm halts after the second step.
\end{proposition}

\begin{theorem}
Let $N = 2^{k_0} \times p_1^{k_1} \times \cdots \times p_l^{k_l}$ with $k_0 \geq 2$, $l \geq 1$, and for all $1 \leq i \leq l$, $p_i \neq 2$. The conditional probability that an element $a \in (\mathbb{Z}/N\mathbb{Z})^*$ of even order satisfies $a^{r/2} \not\equiv -1 \mod N$ is:
\[
\frac{2^{k_0-1}  2^{s_1} \cdots 2^{s_l} - 2}{2^{k_0-1}  2^{s_1} \cdots 2^{s_l} - 1}.
\]
\end{theorem}

\begin{proof}
When $k_0 \geq 3$, the group
\[
(\mathbb{Z}/2^{k_0}\mathbb{Z})^* \times \prod_{i=1}^l (\mathbb{Z}/p_i^{k_i}\mathbb{Z})^*
\]
is isomorphic to
\[
\left( \mathbb{Z}/2\mathbb{Z} \times \mathbb{Z}/2^{k_0-2}\mathbb{Z} \right) \times \prod_{i=1}^l \left( \mathbb{Z}/2^{s_i}\mathbb{Z} \times \mathbb{Z}/m_i\mathbb{Z} \times \mathbb{Z}/p_i^{k_i-1}\mathbb{Z} \right),
\]
where $p_i - 1 = 2^{s_i} m_i$ with $\gcd(m_i,2) = 1$.

We seek the elements of even order $r$ whose half-power $r/2$ corresponds to the image of $-1$, that is:
\[
(1, 0; 2^{s_1 - 1}, 0, 0; \dots; 2^{s_l - 1}, 0, 0).
\]
Let $(x,y; z_1,t_1,u_1; \dots; z_l,t_l,u_l)$ be the components of such an element, with corresponding component-wise orders $(a,b; c_1,d_1,e_1; \dots; c_l,d_l,e_l)$ dividing respectively
\[
(2, 2^{k_0 - 2}; 2^{s_1}, m_1, p_1^{k_1 - 1}; \dots; 2^{s_l}, m_l, p_l^{k_l - 1}).
\]
As the order of an element in a product is the least common multiple (lcm) of the component orders, we get:
\[
r = \mathrm{lcm}(a,b; c_1,d_1,e_1; \dots; c_l,d_l,e_l),
\]
and write $r = 2r'$ with $r'$ odd. Since $a$ must divide $2$ and $a = 1$ would imply $\frac{r}{2}x = 0 \neq 1$, we must have $a = 2$, and $q_a$ (with $r = q_a a$) odd.

Similarly, $b = 1$; if $b = 2$, we would have $y = 2^{k_0 - 3}$, but then $\frac{r}{2}y \neq 0$ as $r/2$ is odd. Hence $y = 0$.

We also require:
\[
q_{c_i} c_i z_i = 0, \quad \text{and} \quad \frac{q_{c_i} c_i}{2} z_i = 2^{s_i - 1}.
\]

The components \( d_i \) and \( e_i \) must be odd, as they are the orders of elements in groups of odd order, and the order of an element divides the order of the group. Consequently, the coefficients \( q_{d_i} \) and \( q_{e_i} \) share the same \( 2 \)-adic valuation, since the total order \( r = q_{c_i} c_i = q_{d_i} d_i \) is even, specifically \( r = 2r' \) with \( r' \) odd.

Conversely, any data of the form
\[
(a = 2, b = 1; c_1 = 2, d_1, e_1; \ldots; c_l = 2, d_l, e_l),
\]
with each triple dividing the corresponding components
\[
(2, 2^{k_0 - 2}; 2^{s_1}, m_1, p_1^{k_1 - 1}; \ldots; 2^{s_l}, m_l, p_l^{k_l - 1}),
\]
defines an element of order
\[
r = \mathrm{lcm}(2, 1; 2, d_1, e_1; \ldots; 2, d_l, e_l),
\]
which is even. Its half-power \( \frac{r}{2} \) corresponds to the element
\[
(1, 0; 2^{s_1 - 1}, 0, 0; \ldots; 2^{s_l - 1}, 0, 0),
\]
which represents the image of \( -1 \) in the group.

For \( c_i = a = 2 \) for all \( i \), such an element
\[
(1, 0; 2^{s_1 - 1}, t_1, u_1; \ldots; 2^{s_l - 1}, t_l, u_l)
\]
has order
\[
r = 2 \cdot \mathrm{lcm}(d_1, e_1; \ldots; d_l, e_l),
\]
which is even. Applying \( \frac{r}{2} \) yields:
\[
(1, 0; 2^{s_1 - 1}, 0, 0; \ldots; 2^{s_l - 1}, 0, 0).
\]
There are exactly
\[
m_1 \cdots m_l  p_1^{k_1 - 1} \cdots p_l^{k_l - 1}
\]
elements of this form.

As established in Step 2, the total number of invertible elements of even order in \( (\mathbb{Z}/N\mathbb{Z})^* \) is:
\[
\left(2^{k_0 - 1}  \prod_{i=1}^l 2^{s_i} - 1 \right)  \prod_{i=1}^l m_i p_i^{k_i - 1}.
\]
Therefore, the conditional probability that an element of even order has half-power not equal to \( -1 \) is:
\[
\frac{2^{k_0 - 1}  2^{s_1} \cdots 2^{s_l} - 2}{2^{k_0 - 1}  2^{s_1} \cdots 2^{s_l} - 1}.
\]

When \( k_0 = 2 \), the constraint on the \( y \)-component vanishes, but the enumeration remains unchanged. Thus, the formula also holds in the case \( k_0 = 2 \).
\end{proof}
\begin{theorem}
Let \( N = p_1^{k_1} \times \cdots \times p_l^{k_l} \) or \( N = 2 p_1^{k_1} \times \cdots \times p_l^{k_l} \), with \( l \geq 1 \) and all \( p_i \neq 2 \). Then, the conditional probability that an element in \( (\mathbb{Z}/N\mathbb{Z})^* \) of even order has a half-power different from \( -1 \) is given by:
\[
\frac{2^{s_1} \cdots 2^{s_l} - \dfrac{2^{ls} - 1}{2^l - 1} - 1}{2^{s_1} \cdots 2^{s_l} - 1}.
\]
\end{theorem}

\begin{proof}
The group
\[
(\mathbb{Z}/2\mathbb{Z})^* \times \prod_{i=1}^l (\mathbb{Z}/p_i^{k_i}\mathbb{Z})^*
\quad \text{or} \quad
\prod_{i=1}^l (\mathbb{Z}/p_i^{k_i}\mathbb{Z})^*
\]
is isomorphic to
\[
\prod_{i=1}^l \left( \mathbb{Z}/2^{s_i}\mathbb{Z} \times \mathbb{Z}/m_i\mathbb{Z} \times \mathbb{Z}/p_i^{k_i - 1}\mathbb{Z} \right),
\]
where \( p_i - 1 = 2^{s_i} m_i \) and \( \gcd(m_i, 2) = 1 \).

We seek elements \( (z_1, t_1, u_1; \dots; z_l, t_l, u_l) \) of even order \( r \) such that:
\[
\frac{r}{2}(z_1, t_1, u_1; \dots; z_l, t_l, u_l) = (2^{s_1 - 1}, 0, 0; \dots; 2^{s_l - 1}, 0, 0),
\]
which corresponds to the image of \( -1 \) under the isomorphisms.

Let \( (c_1, d_1, e_1; \dots; c_l, d_l, e_l) \) be the orders of the components, dividing:
\[
(2^{s_1}, m_1, p_1^{k_1 - 1}; \dots; 2^{s_l}, m_l, p_l^{k_l - 1}).
\]
The overall order is the least common multiple:
\[
r = \mathrm{lcm}(c_1, d_1, e_1; \dots; c_l, d_l, e_l).
\]
Now, \( q_{c_i} c_i z_i = 0 \) and \( \frac{q_{c_i} c_i}{2} z_i = 2^{s_i - 1} \) imply that \( q_{c_i} \) is odd and \( c_i \) is a power of 2. All \( c_i \) must be equal, say \( c_i = 2^\alpha \), with \( 1 \leq \alpha \leq s := \min_i(s_i) \). Otherwise, the half-power condition fails.

Since \( d_i \) and \( e_i \) are from groups of odd order, they must also be odd. Therefore, their respective \( q_{d_i} \), \( q_{e_i} \) values share the same power of 2 as \( q_{c_i} \), ensuring \( r \) is even.

Conversely, any such tuple with \( c_i = c = 2^\alpha \), \( 1 \leq \alpha \leq s \), and suitable \( d_i, e_i \) defines an element of even order \( r \), such that
\[
\frac{r}{2}(z_1, t_1, u_1; \dots; z_l, t_l, u_l) = (2^{s_1 - 1}, 0, 0; \dots; 2^{s_l - 1}, 0, 0).
\]

The number of such elements depends on the number of choices for \( z_i \). For a fixed order \( 2^n \) with \( 1 \leq n \leq s \), the number of elements of order \( 2^n \) in \( \mathbb{Z}/2^{s_i}\mathbb{Z} \) is \( \varphi(2^n) = 2^{n-1} \). Hence, for each \( n \), we get \( 2^{l(n-1)} \) such tuples across \( l \) components.

Summing over all possible orders gives:
\[
\sum_{j=0}^{s-1} 2^{lj} = \frac{2^{ls} - 1}{2^l - 1}.
\]

As the components \( t_i \), \( u_i \) are unconstrained, the total number of such elements is:
\[
\frac{2^{ls} - 1}{2^l - 1}  \prod_{i=1}^l m_i p_i^{k_i - 1}.
\]

From Step 2, the total number of elements of even order is:
\[
\left( \prod_{i=1}^l 2^{s_i} m_i p_i^{k_i - 1} \right) - \left( \prod_{i=1}^l m_i p_i^{k_i - 1} \right).
\]

Therefore, the desired conditional probability is:
\[
\frac{\prod_{i=1}^l 2^{s_i} m_i p_i^{k_i - 1} - \prod_{i=1}^l m_i p_i^{k_i - 1} - \dfrac{2^{ls} - 1}{2^l - 1}  \prod_{i=1}^l m_i p_i^{k_i - 1}}{\prod_{i=1}^l 2^{s_i} m_i p_i^{k_i - 1} - \prod_{i=1}^l m_i p_i^{k_i - 1}} = \frac{2^{s_1} \cdots 2^{s_l} - \dfrac{2^{ls} - 1}{2^l - 1} - 1}{2^{s_1} \cdots 2^{s_l} - 1}.
\]

\end{proof}

\begin{corollary}
\leavevmode
\begin{enumerate}[label=(\alph*)]
\item For every \( N \geq 2 \), the success probability of Step 1 of Shor's algorithm is non-zero.
\item Shor’s algorithm fails at Step 2 if and only if \( N = 2 \).
\item Shor’s algorithm fails at Step 3 if and only if \( N = 4 \), or \( N = 2p^k \), or \( N = p^k \), for some prime \( p \neq 2 \), \( k \geq 1 \).
\end{enumerate}
\end{corollary}

\begin{proof}
\leavevmode
\begin{enumerate}[label=(\alph*)]
\item This follows directly from the fact that \( \varphi(N) \neq 0 \) for all \( N \geq 2 \).

\item Consider the conditional probabilities obtained in Step 2. For all \( s_i \geq 1 \),
\[
1 - \frac{1}{2^{s_1}} \cdots \frac{1}{2^{s_r}} \neq 0.
\]
The case where the probability becomes zero is:
\[
1 - \frac{1}{2^{k_0 - 1}} \prod_{i=1}^l \frac{1}{2^{s_i}} = 0 \iff k_0 = 1 \text{ and all } s_i = 0,
\]
which implies \( N = 2 \).

\item For \( N = 2^{k_0}  p_1^{k_1} \cdots p_l^{k_l} \), with \( k_0 \geq 2 \), the conditional probability that an element has half-power different from \( -1 \), given that it is coprime to \( N \) and of even order, is:
\[
\frac{2^{k_0 - 1}  2^{s_1} \cdots 2^{s_l} - 2}{2^{k_0 - 1} 2^{s_1} \cdots 2^{s_l} - 1}.
\]
This is zero if and only if \( 2^{k_0 - 1}  2^{s_1} \cdots 2^{s_l} - 2 = 0 \), i.e., \( k_0 = 2 \) and \( l = 0 \), hence \( N = 4 \).

Now consider \( N = p_1^{k_1} \cdots p_l^{k_l} \) or \( N = 2p_1^{k_1} \cdots p_l^{k_l} \). The conditional probability is:
\[
\frac{2^{s_1} \cdots 2^{s_l} - \frac{2^{ls} - 1}{2^l - 1} - 1}{2^{s_1} \cdots 2^{s_l} - 1}.
\]
Let’s study when the numerator vanishes:
\[
2^{s_1} \cdots 2^{s_l} - 1 - \sum_{j=0}^{s-1} 2^{jl} = 0.
\]
This holds when \( l = 1 \) and \( s = s_1 \), because:
\[
2^s - 1 - \sum_{j=0}^{s-1} 2^j = 2^s - 1 - (2^s - 1) = 0.
\]
So this occurs when \( N = p^k \) or \( N = 2p^k \), with \( p \neq 2 \) and \( k \geq 1 \).

For \( l \geq 2 \), if \( s = 1 \), then:
\[
2 \cdots 2 - 1 - 1 = 0 \quad \text{has no solution}.
\]
For \( l \geq 2 \), \( s \geq 2 \), the equality becomes:
\[
2^{s_1} \cdots 2^{s_l} = 2 + 2^l + \cdots + 2^{l(s-1)},
\]
but the left-hand side is divisible by 4 while the right-hand side is not, so equality is impossible.

Hence, Shor's algorithm fails at Step 3 exactly when \( N = p^k \) or \( N = 2p^k \) with \( p \neq 2 \).
\end{enumerate}
\end{proof}

\section{Success Probability of Shor’s Algorithm}\label{success-pr}

\begin{theorem}
Let \( N = 2^{k_0} p_1^{k_1} \cdots p_l^{k_l} \), with \( l \geq 0 \), \( p_i \neq 2 \) primes, and \( p_i - 1 = 2^{s_i} m_i \), \( k_0 \geq 2 \). Then the overall success probability of Shor’s algorithm is:
\[
\frac{m_1 \cdots m_l}{2 p_1 \cdots p_l} \times \frac{2^{k_0 - 1}  2^{s_1} \cdots 2^{s_l} - 2}{2^{k_0 - 1}}.
\]
\end{theorem}
\begin{proof}
The result is obtained by multiplying three quantities:
\begin{itemize}
    \item the probability that a uniformly chosen integer \( a < N \) is coprime to \( N \), which is \( \dfrac{2^{s_1} \cdots 2^{s_l} m_1 \cdots m_l}{2 p_1 \cdots p_l} \);
    \item the conditional probability that \( a \in (\mathbb{Z}/N\mathbb{Z})^* \) has even order, given it is coprime to \( N \), which is \( 1 - \dfrac{1}{2^{k_0 - 1} 2^{s_1} \cdots 2^{s_l}} \);
    \item the conditional probability that \( a \) has half-order power not equal to \( -1 \), given it is coprime to \( N \) and has even order, which is:
    \[
    \dfrac{2^{k_0 - 1} 2^{s_1} \cdots 2^{s_l} - 2}{2^{k_0 - 1} 2^{s_1} \cdots 2^{s_l} - 1}.
    \]
\end{itemize}

Thus, the overall success probability becomes:
\[
\frac{2^{s_1} \cdots 2^{s_l} m_1 \cdots m_l}{2 p_1 \cdots p_l} \times \frac{2^{k_0 - 1} 2^{s_1} \cdots 2^{s_l} - 2}{2^{k_0 - 1} 2^{s_1} \cdots 2^{s_l}} = \frac{m_1 \cdots m_l}{2 p_1 \cdots p_l} \times \frac{2^{k_0 - 1} 2^{s_1} \cdots 2^{s_l} - 2}{2^{k_0 - 1}}.
\]

When \( l = 0 \), we recover the expression previously established for \( N = 2^k \).
\end{proof}

\begin{theorem}
Let \( N = p_1^{k_1} \cdots p_l^{k_l} \) or \( N = 2p_1^{k_1} \cdots p_l^{k_l} \), with \( l \geq 1 \), and \( p_i \neq 2 \) distinct primes such that \( p_i - 1 = 2^{s_i} m_i \) with \( \gcd(2, m_i) = 1 \). Then the success probability of Shor’s algorithm is:
\[
\frac{m_1 \cdots m_l}{p_1 \cdots p_l} \times \left(2^{s_1} \cdots 2^{s_l} - \frac{2^{ls} - 1}{2^l - 1} - 1 \right).
\]
\end{theorem}

\begin{proof}
We again multiply the following:
\begin{itemize}
    \item The probability that \( a < N \) is coprime to \( N \): \( \dfrac{2^{s_1} \cdots 2^{s_l} m_1 \cdots m_l}{p_1 \cdots p_l} \);
    \item The conditional probability that \( a \) has even order: \( \dfrac{2^{s_1} \cdots 2^{s_l} - 1}{2^{s_1} \cdots 2^{s_l}} \);
    \item The conditional probability that \( a^{r/2} \neq -1 \), given even order: 
    \[
    \frac{2^{s_1} \cdots 2^{s_l} - \dfrac{2^{ls} - 1}{2^l - 1} - 1}{2^{s_1} \cdots 2^{s_l} - 1}.
    \]
\end{itemize}

Putting everything together:
\[
\frac{2^{s_1} \cdots 2^{s_l} m_1 \cdots m_l}{p_1 \cdots p_l} \times \frac{2^{s_1} \cdots 2^{s_l} - 1}{2^{s_1} \cdots 2^{s_l}} \times \frac{2^{s_1} \cdots 2^{s_l} - \dfrac{2^{ls} - 1}{2^l - 1} - 1}{2^{s_1} \cdots 2^{s_l} - 1} \]
\[=
\frac{m_1 \cdots m_l}{p_1 \cdots p_l} \times \left(2^{s_1} \cdots 2^{s_l} - \frac{2^{ls} - 1}{2^l - 1} - 1\right).
\]
\end{proof}
\section{Simulation Procedure}\label{sec:pseudocode}

We consider a setting where the factorization of \( N \) is known. The success probability of a run of Shor’s algorithm depends on three quantities: the probability that a uniformly chosen integer is coprime to \( N \), the probability that it has even order modulo \( N \), and the probability that its half-power is not congruent to \( -1 \mod N \). These expressions have been derived in the previous sections.

The following pseudocode defines a simulation of one run of Shor’s algorithm using these probabilities. It assumes \( N = 2^{k_0} p_1^{k_1} \cdots p_l^{k_l} \), with \( p_i \) odd primes such that \( p_i - 1 = 2^{s_i} m_i \), and returns a Boolean indicating whether the run would succeed.
This procedure is intended for use in test environments such as algorithm benchmarking or quantum resource estimation.

\begin{algorithm}[H]
\caption{Simulated Shor Success}\label{algo:simulate-shor}
\begin{algorithmic}[1]
\Require Integer \( N \geq 2 \) with known factorization \( N = 2^{k_0} p_1^{k_1} \cdots p_l^{k_l} \), where \( p_i \neq 2 \) and \( p_i - 1 = 2^{s_i} m_i \)
\Ensure Boolean value: \textbf{True} if simulated run succeeds, \textbf{False} otherwise

\State Compute \( \varphi(N) = \prod (p_i - 1)p_i^{k_i - 1} \)
\State \( \Pr_{\text{gcd}} \gets \varphi(N) / N \)

\If{ \( k_0 \geq 2 \) }
    \State \( \Pr_{\text{even}} \gets 1 - \frac{1}{2^{k_0 - 1}  2^{s_1} \cdots 2^{s_l}} \)
    \State \( \Pr_{\text{non-}(-1)} \gets \frac{2^{k_0 - 1}  2^{s_1} \cdots 2^{s_l} - 2}{2^{k_0 - 1}  2^{s_1} \cdots 2^{s_l} - 1} \)
\Else
    \State \( \Pr_{\text{even}} \gets 1 - \frac{1}{2^{s_1} \cdots 2^{s_l}} \)
    \State \( s \gets \min(s_1, \ldots, s_l) \)
    \State \( T \gets \frac{2^{ls} - 1}{2^l - 1} \)
    \State \( \Pr_{\text{non-}(-1)} \gets \frac{2^{s_1} \cdots 2^{s_l} - T - 1}{2^{s_1} \cdots 2^{s_l} - 1} \)
\EndIf

\State \( \Pr_{\text{success}} \gets \Pr_{\text{gcd}} \cdot \Pr_{\text{even}} \cdot \Pr_{\text{non-}(-1)} \)
\State Draw \( r \sim \text{Uniform}[0,1] \)
\If{ \( r < \Pr_{\text{success}} \) }
    \State \Return \textbf{True}
\Else
    \State \Return \textbf{False}
\EndIf
\end{algorithmic}
\end{algorithm}

\section{Conclusion}\label{conclusion}

We have computed exact expressions for the probability of success of Shor’s algorithm in terms of the prime factorization of \( N \). The analysis identifies the structure of \( (\mathbb{Z}/N\mathbb{Z})^* \) as the determining factor in whether the algorithm succeeds, fails, or does so with fixed probability. All cases where failure is deterministic have been identified.

These results yield a closed-form success probability, which can be used in simulations where the factorization is known. This setting is standard in algorithm validation and quantum simulation benchmarks. The proposed pseudocode computes the success probability and simulates one run of the algorithm under ideal conditions.

\end{document}